\documentclass[letterpaper, 10 pt, conference]{ieeeconf}  

\usepackage[utf8]{inputenc}
\usepackage{amsfonts}
\usepackage{amsmath}
\usepackage{comment}
\usepackage{xcolor}
\usepackage{algorithm, algpseudocode}
\usepackage{subcaption}

\usepackage{cite}

\usepackage{pgfplots}
  \pgfplotsset{compat=newest}
  \usetikzlibrary{plotmarks}
  \usetikzlibrary{arrows.meta}
  \usepgfplotslibrary{patchplots}
  \usepackage{grffile}

\newcommand{\R}{\mathbb{R}}
\newcommand{\bmat}[1]{\begin{bmatrix}#1\end{bmatrix}}
\newcommand\norm[1]{\left\lVert#1\right\rVert}

\newtheorem{theorem}{Theorem}
\newtheorem{definition}{Definition}

\newtheorem{example}{Example}
\newtheorem{lemma}{Lemma}
\newtheorem{assumption}{Assumption}

\IEEEoverridecommandlockouts                              
\title{\LARGE \bf
Backward Reachability using Integral Quadratic Constraints for Uncertain Nonlinear Systems
}

\author{He Yin, Peter Seiler and  Murat Arcak
\thanks{Funded in part by the Air Force Office of Scientific Research grant FA9550-18-1-0253, the Office of Naval Research grant N00014-18-1-2209, and  the U.S. National
	Science Foundation grant ECCS-1906164.}
\thanks{H. Yin is with the Department of Mechanical Engineering, University of California, Berkeley {\tt\small he\_yin@berkeley.edu.}}
\thanks{P. Seiler is with the Department of Electrical Engineering and Computer Science, University of Michigan,  Ann Arbor {\tt\small pseiler@umich.edu.}}
\thanks{M. Arcak is with the Department of Electrical Engineering and Computer Sciences, University of California, Berkeley {\tt\small arcak@eecs.berkeley.edu.}}%
}

\begin{document}

\maketitle
\thispagestyle{empty}
\pagestyle{empty}

\begin{abstract}
A method is proposed to compute robust inner-approximations to the backward reachable set for uncertain nonlinear systems. It also produces a robust control law that drives trajectories starting in these sets to the target set. The method merges dissipation inequalities and integral quadratic constraints (IQCs) with both hard and  soft IQC factorizations.  Computational algorithms are presented using the generalized S-procedure and  sum-of-squares techniques.  The use of IQCs in backward reachability  analysis allows for a  variety of perturbations including parametric uncertainty,  unmodeled dynamics, nonlinearities, and uncertain time delays. The method is demonstrated on two examples, including a 6-state quadrotor with actuator uncertainties.
\end{abstract}

\section{Introduction}
The backward reachable set (BRS) is the set of initial conditions whose successors can be driven to the target set at the end of a finite time horizon with an admissible controller. The BRS is of vital importance for safety-critical systems, since it provides a safe envelope for the system to reach the target set and avoid obstacles \cite{MoChen:16}. 

Backward reachability has been studied with several approaches. Occupation measure-based methods \cite{Henrion:14, Majumdar:14, Kousik:16} compute BRS outer-approximations, but do not guarantee reaching the target set. In contrast, the exact BRS is computed in \cite{Mitchell:00, Mitchell:05, Darbon:2016, donggun2019CDC} as the sublevel set of the solution to Hamilton-Jacobi (HJ) partial differential equations (PDEs). Other results provide BRS inner-approximations using relaxed HJ equations \cite{xue2019inner, xue2019robust, Jones:19} and Lyapunov-based methods \cite{Anirudha:13}. 

A shortcoming of the existing  reachability tools is that they rely on accurate system models. Only limited forms of uncertainty have been addressed, such as parametric uncertainty in \cite{Mitchell:00, Kousik:16, xue2019inner, xue2019robust, Jones:19, Anirudha:13} and both parametric uncertainty and $\mathcal{L}_2$ disturbances in our earlier work  \cite{Yin:ACC, Yin:TAC}.

In this paper, we propose a method to compute inner-approximations to the BRS that are robust to a more general class of perturbations. 
We model the uncertain nonlinear system as an interconnection of the nominal system $G$ and the perturbation $\Delta$, as in Fig.~1. The input-output relationship of $\Delta$ is described using the integral quadratic constraint (IQC) framework \cite{Megretski:97, Veenman:16},  which accounts for parametric uncertainties, unmodeled dynamics, slope-bounded nonlinearities, and uncertain  time delays. 
We characterize BRS inner-approximations by sublevel sets of storage functions that satisfy a dissipation inequality that is compatible with IQCs. We derive an algorithm to compute storage functions and associated control laws using the generalized S-procedure \cite{Parrilo:00} and SOS techniques \cite{Antonis:02, Jarvis:05}. These techniques allow us to formulate iterative convex optimization procedures for the computation of storage functions and control laws.

The specific contributions of this paper are threefold. First, we propose a general framework for robust backward reachability of uncertain nonlinear systems, allowing for various types of uncertainty beyond parametric uncertainty. 
Second, we incorporate both hard and soft IQC factorizations in the framework.  The use of dissipation inequalities typically requires IQCs that are valid over any finite time horizon, known as hard IQCs.  However, many IQCs are specified in the frequency domain, which are equivalent to time-domain constraints over  infinite horizons (soft IQCs). We obtain improved BRS bounds by incorporating soft IQCs  by means of the finite-horizon bound derived in \cite{Veenman:16}. Third, we overcome a technical challenge that arises when the input of the perturbation $\Delta$ depends directly on the control command, as in the case of actuator uncertainty. This dependence creates a source of nonconvexity, which we circumvent by introducing auxiliary states in the control law. 

The paper is organized as follows. Section~\ref{sec:hardIQC} presents the problem setup, and the robust backward reachability framework using hard IQCs. The method is adapted to actuator uncertainties in Section~\ref{sec:actuator_uncert}. Section~\ref{sec:softIQC} extends the robust reachability analysis to soft IQCs. Two examples, including a 6-state quadrotor system with actuator uncertainty, are given in Section~\ref{sec:Numeric_ex}. Section~\ref{sec:conclusion} summarizes the results.
\subsection{Notation}
$\mathbb{R}^{m\times n}$ and $\mathbb{S}^n$ denote the set of $m$-by-$n$ real matrices and $n$-by-$n$ real, symmetric matrices. 
$\mathbb{RL}_{\infty}$ is the set of rational functions with real coefficients that have no poles on the imaginary axis. $\mathbb{RH}_{\infty} \subset \mathbb{RL}_{\infty}$ contains functions that are analytic in the closed right-half of the complex plane. $\mathcal{L}^{n_r}_2$ is the space of measureable functions $r: [0, \infty) \rightarrow \R^{n_r}$ with $\norm{r}_2^2 := \int_0^\infty r(t)^\top r(t)dt < \infty$. Define the finite-horizon $\mathcal{L}_2$ norm as $\norm{r}_{2,[0,T]}:= \left(\int_{0}^T r(t)^\top r(t)dt \right)^{1/2}$. If $\norm{r}_{2,[0,T]} < \infty$ then $r \in \mathcal{L}^{n_r}_2[0,T]$. The finite horizon induced $\mathcal{L}_2$ to $\mathcal{L}_2$ norm of an operator is denoted as $\norm{\cdot}_{2 \rightarrow 2, [0,T]}$.
For $\xi \in \mathbb{R}^n$, $\mathbb{R}[\xi]$ represents the set of polynomials in $\xi$ with real coefficients, and $\R^m[\xi]$ and $\R^{m \times p}[\xi]$ denote all vector and matrix valued polynomial functions. The subset $\Sigma[\xi] := \{\pi = \sum_{i=1}^M \pi_i^2: \pi_1,...,\pi_M \in \mathbb{R}[\xi]\}$ of $\mathbb{R}[\xi]$ is the set of SOS polynomials in $\xi$. 
For $\eta \in \mathbb{R}$, and continuous $g : \mathbb{R} \times \mathbb{R}^{n} \rightarrow \mathbb{R}$, define $\Omega_{t,\eta}^{g} := \{x \in \mathbb{R}^n : g(t,x) \le \eta\}$, a $t$-dependent set. $KYP$ denotes a mapping to the block 2-by-2 matrix: $KYP(Y,A,B,C,D,M) := 
 \left[\begin{smallmatrix}A^\top Y + YA & YB \\ B^\top Y & 0\end{smallmatrix}\right] + \left[\begin{smallmatrix}\star \end{smallmatrix}\right]^\top M \left[\begin{smallmatrix}C & D\end{smallmatrix}\right]. $

\section{Backward Reachability with Hard IQCs}\label{sec:hardIQC}
\subsection{Problem Setup} \label{sec:prelim}
Consider the following uncertain nonlinear system:
\begin{subequations} \label{eq:uncertain_sys}
\begin{align}
\dot{x}_G(t) &= f(x_G(t),w(t),d(t))  + g(x_G(t),w(t),d(t))u(t), \\
v(t) &= h(x_G(t),w(t),d(t)), \label{eq:expression_v} \\
w(\cdot) &= \Delta(v(\cdot)), 
\end{align}
\end{subequations}
which is an interconnection (Fig.~\ref{fig:Fu}) of the nominal system $G$ and the perturbation $\Delta$, denoted as $F_u(G,\Delta)$. In \eqref{eq:uncertain_sys}, $x_G(t) \in \R^{n_G}$ is the state, $u(t) \in U \subseteq \R^{n_u}$ is the control input, $d(t) \in \R^{n_d}$ is the external disturbance, and $v(t) \in \R^{n_v}$ and $w(t) \in \R^{n_w}$ are the inputs and outputs of $\Delta$. The mappings $f: \R^{n_G} \times \R^{n_w} \times \R^{n_d} \rightarrow \R^{n_G}, g: \R^{n_G} \times \R^{n_w} \times \R^{n_d}\rightarrow \R^{n_G \times n_u}$, and $h: \R^{n_G} \times \R^{n_w} \times \R^{n_d}\rightarrow \R^{n_v}$ define the nominal system $G$. The perturbation $\Delta: \mathcal{L}_{2}^{n_v} \rightarrow \mathcal{L}_{2}^{n_w}$ is an operator. Note that in \eqref{eq:expression_v}, $v$ does not depend directly on $u$.
\begin{figure}[h]
	\centering
	\includegraphics[width=0.2\textwidth]{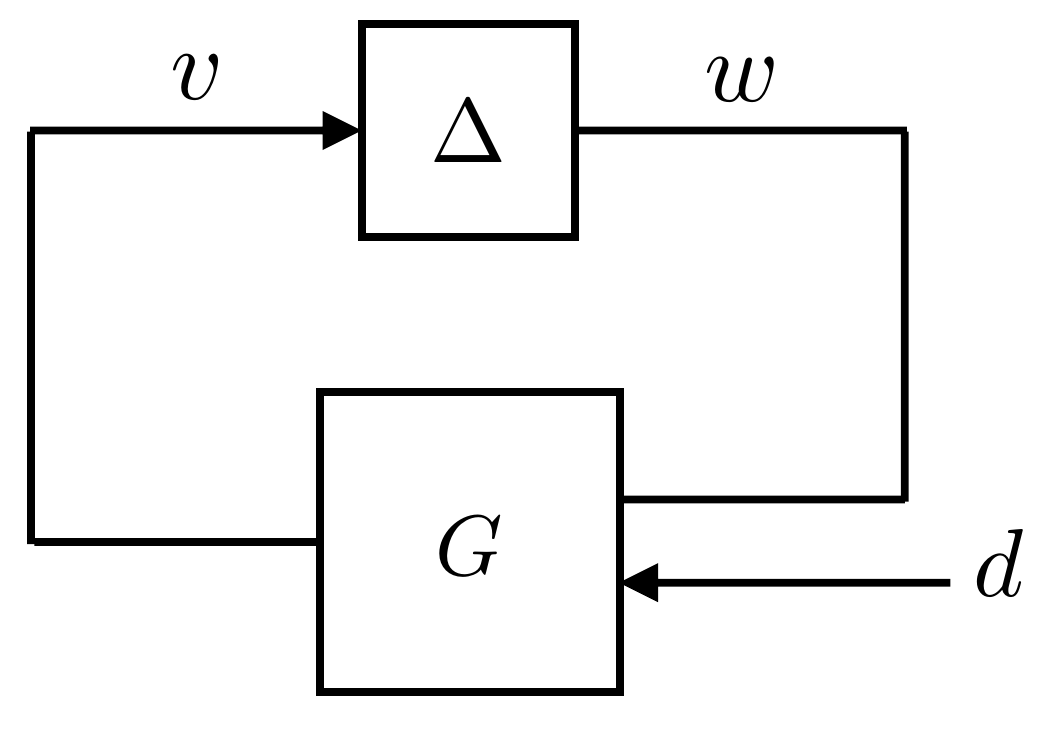}
	\caption{Interconnection $F_u(G,\Delta)$ of $G$ and $\Delta$} 
	\label{fig:Fu}    
\end{figure}
\begin{assumption} \label{ass:disturb}
(i) $d$ has bounded $\mathcal{L}_2$ energy: 
\begin{align}
\norm{d}_{2,[0,T]} < R, \ \text{with} \ R >0, \ \text{and}\label{eq:w_ass}
\end{align}
(ii) the set of control constraints is given as a polytope $U:= \{u \in \R^{n_u}: P u \leq b \}$, where $P \in \R^{n_p \times  n_u}$ and $b \in \R^{n_p}$.
\end{assumption}

Let $x_G(t;\xi,u,d)$ define the solution to the uncertain system \eqref{eq:uncertain_sys}, at time $t$ $(0 \leq t \leq T)$, from the initial condition $\xi$, under the control $u$ and the disturbance $d$. The definition of the backward reachable set (BRS) is given as follows.
\begin{definition}
	Under Assumption~\ref{ass:disturb}, the BRS of $F_u(G,\Delta)$ \eqref{eq:uncertain_sys} is defined as $BRS(T,X_T,U,R,F_u(G,\Delta)):=$
	\begin{align}
	& \{\xi \in \R^{n_G}: \exists u, \ \text{s.t.} \ u(t) \in U \ \forall t \in [0,T], \ \text{and}  \nonumber \\
	&~~~~~~~~~~~~~~ x_G(T;\xi,u,d) \in X_T \  \forall d \ \text{with} \ \norm{d}_{2,[0,T]} < R\}. \nonumber
	\end{align}
\end{definition}

The goal of this paper is to compute an inner-approximation to the BRS and an associated controller that certifies the inner-approximation.

\subsection{Integral Quadratic Constraints}
The perturbation $\Delta$ can represent various types of uncertainties and nonlinearities, including parametric uncertainty, unmodeled dynamics, slope-bounded nonlinearities, and uncertain time delays \cite{Megretski:97, Veenman:16}. To characterize $\Delta$ with an integral quadratic constraint (IQC) we apply a `virtual' filter $\Psi$ to the input $v$ and output $w$ of $\Delta$, as illustrated in Fig.~\ref{fig:filter}, and impose quadratic constraints on the output $z$ of $\Psi$. 
\begin{figure}[h]
	\centering
	\includegraphics[width=0.3\textwidth]{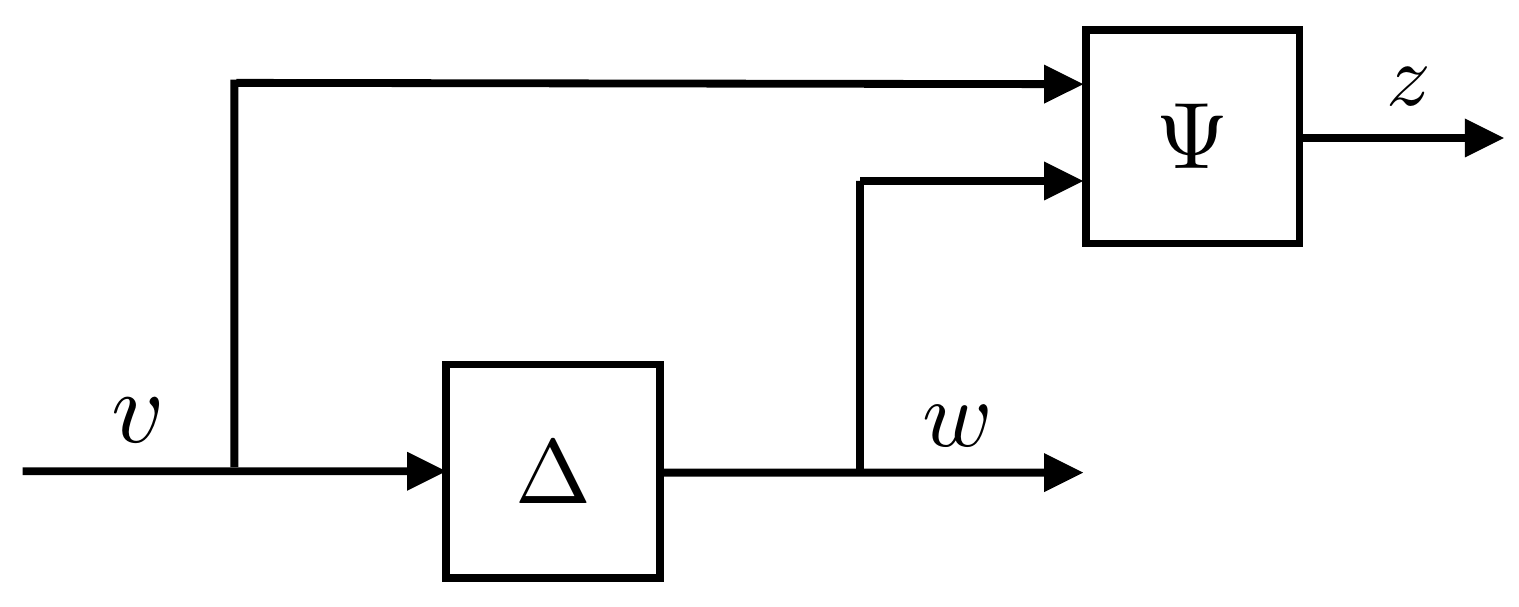}
	\caption{Pictorial illustration of the filter $\Psi$}
	\label{fig:filter}    
\end{figure}
The filter $\Psi$ is an LTI system driven by $(v,w)$, with zero initial condition $x_\psi(0) = 0$, and dynamics of the form: 
\begin{subequations}
	\begin{align}
	\dot{x}_{\psi}(t) &= A_\psi x_\psi(t) + B_{\psi 1}v(t) + B_{\psi 2}w(t), \\
	z(t) &= C_\psi x_\psi(t) + D_{\psi 1}v(t) + D_{\psi 2}w(t), \label{eq:filter_output}
	\end{align}
\end{subequations}
where $x_\psi(t) \in \R^{n_\psi}$ is the state, and $z(t) \in \R^{n_z}$ is the output. IQCs can be defined in both the time and frequency domain. The use of time domain IQCs is required for the dissipation-type results used later in this paper. Time domain IQCs consist of hard IQCs and soft IQCs, which are quadratic constraints on $z$ over finite and infinite horizons, respectively. In this section, we focus on the analysis with hard IQCs. 
\begin{definition}
	Given $\Psi \in \R\mathbb{H}_{\infty}^{n_z \times (n_v + n_w)}$ and $M \in \mathbb{S}^{n_z}$. A bounded, causal operator $\Delta : \mathcal{L}_{2}^{n_v} \rightarrow \mathcal{L}_{2}^{n_w}$ satisfies the \underline{hard IQC} defined by $(\Psi,M)$ if, for all $v \in \mathcal{L}_{2}^{n_v}$, and $w = \Delta(v)$,
	\begin{align}
	\int_0^t z(\tau)^\top M z(\tau) d\tau \ge 0, \forall t \in [0,T].
	\end{align}
\end{definition}
 
The notation $\Delta \in$ HardIQC$(\Psi,M)$ indicates that $\Delta$ satisfies the hard IQC defined by ($\Psi$, $M$). The following example gives two types of $\Delta$ and corresponding hard IQCs:
\begin{example}\label{ex:hardIQC}
(a) Consider the set of LTI uncertainties with a given norm bound $\sigma > 0$: $\Delta \in \R\mathbb{H}_\infty$ with $\norm{\Delta}_\infty \leq \sigma$. It is proved in \cite{Balakrishnan:02} that $\Delta \in$ HardIQC$(\Psi, M_D)$, where $\Psi = \left[\begin{smallmatrix}\Psi_{11} & 0\\ 0 & \Psi_{11}\end{smallmatrix}\right]$ with $\Psi_{11} \in \R\mathbb{H}_\infty^{n_z \times 1}$ and 
\begin{align}
M_D \in \mathcal{M}_1 := \left\{\left[\begin{smallmatrix}\sigma^2 M_{11} & 0\\ 0 & - M_{11}\end{smallmatrix}\right]: M_{11} \succeq 0 \right\}. 
\end{align}
A typical choice for $\Psi_{11}$ \cite{Veenman:16} is 
\begin{align}\label{eq:filter_choice}
\Psi_{11}^{d,m} = \bmat{ 1, \frac{1}{(s+m)},...,\frac{1}{(s+m)^d} }^\top, \ \text{with} \ m > 0,
\end{align}
where $m$ and $d$ are selected by the user.

(b) Consider the set of nonlinear, time varying, uncertainties with a given norm-bound $\sigma$: $\norm{\Delta}_{2 \rightarrow 2, [0,T]} \leq \sigma$. $\Delta$ satisfies the hard IQCs defined by $\Psi = I_{n_v + n_w}$ and 
\begin{align}\label{eq:M_nonlinearity}
M \in \mathcal{M}_2 := \left\{ \left[\begin{smallmatrix} \sigma^2 \lambda I_{n_v} & 0 \\ 0 & -\lambda I_{n_w} \end{smallmatrix}\right]: \lambda \ge 0\right\}.
\end{align}
\end{example}

\subsection{Robust Backward Reachability}
As illustrated in the previous examples, each type of $\Delta$ can be characterized by corresponding hard IQCs associated with a filter $\Psi$ and a matrix $M$. The analysis on $F_u(G,\Delta)$ can be instead performed on the extended system shown in Fig.~\ref{fig:ExtendSys}, with an additional constraint $\Delta \in$ HardIQC$(\Psi,M)$. The extended system is an interconnection of $G$ and $\Psi$, with combined state vector $x: = [x_G; x_\psi]\in \R^{n}$, $n = n_G + n_\psi$, whose dynamics can be rewritten as 
\begin{subequations}\label{eq:extend_sys}
\begin{align}
\dot{x}(t) &= F(x(t), w(t), d(t), u(t)),\label{eq:extended_F} \\
z(t) &= H(x(t), w(t), d(t)), \label{eq:z_expression}
\end{align}
\end{subequations}
where $F: \R^{n} \times \R^{n_w} \times \R^{n_d} \times \R^{n_u} \rightarrow \R^{n}$ and $H: \R^{n} \times \R^{n_w} \times \R^{n_d} \rightarrow \R^{n_z}$ depend on the dynamics of $G$ and $\Psi$. $F$ is still affine in $u$.
\begin{figure}[h]
	\centering
	\includegraphics[width=0.25\textwidth]{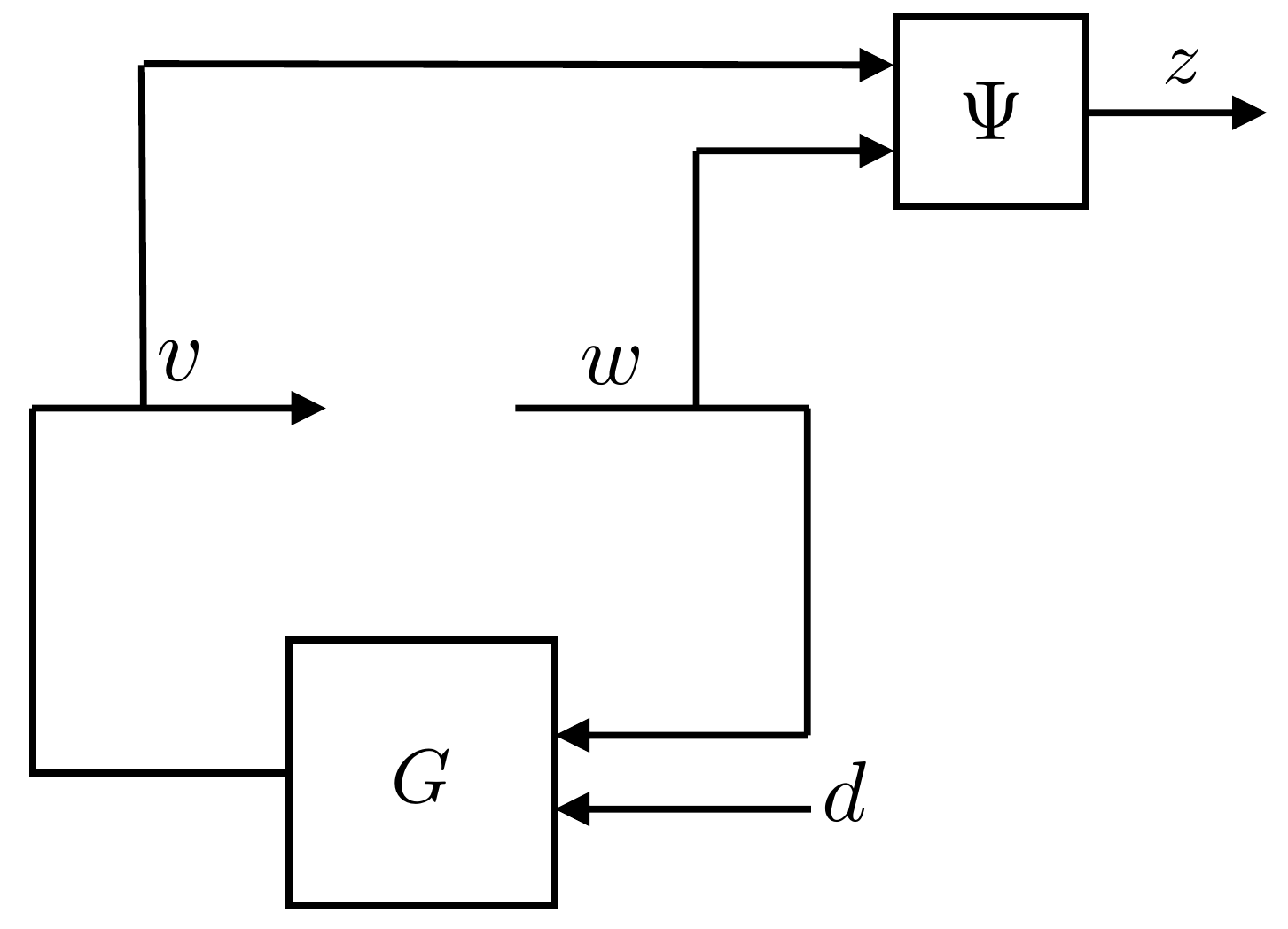}
	\caption{Extended system of $G$ and $\Psi$}
	\label{fig:ExtendSys}    
\end{figure}

We consider the memoryless, time-varying state-feedback control $u(t) = k(t,x_G(t))$, $k: \R \times \R^{n_G} \rightarrow \R^{n_u}$. We don't allow $k$ to depend on $x_\psi$, since $x_\psi$ is introduced by the virtual filter $\Psi$.  The following theorem provides a BRS inner-approximation for the extended system $G$ and $\Psi$, and therefore for the original uncertain system $F_u(G,\Delta)$, with control $k$.
\begin{theorem}\label{thm:hardIQC}
Let Assumption~\ref{ass:disturb} hold, and assume $\Delta \in$ HardIQC$(\Psi, M)$, with $\Psi$ and $M$ given. Given $X_T \subset \R^{n_G}$, $P \in \R^{n_p \times  n_u}$, $b \in \R^{n_p}$, $R > 0$, $F, H$ defined in \eqref{eq:extend_sys}, $T > 0$, and $\gamma \in \R$, if there exists a $\mathcal{C}^1$ function $V: \R \times \R^{n} \rightarrow \R$, and a control law $k: \R \times \R^{n_G} \rightarrow \R^{n_u}$ that is continuous in $t$ and locally Lipschitz in $x_G$, such that
\begingroup
\allowdisplaybreaks
\begin{align}
&\partial_t V(t,x) + \partial_{x} V(t,x)\cdot F(x, w, d, k) + z^\top M z \leq d^\top d, \nonumber \\
& \quad \quad  \forall (t,x,w,d) \in  [0,T]\times \R^{n} \times \R^{n_w} \times \R^{n_d}, \nonumber  \\
& \quad \quad \text{s.t.} \ V(t, x) \leq \gamma + R^2, \label{eq:diss_basic} \\
& \{x_G : V(T,x)\leq \gamma + R^2 \} \subseteq X_T, \ \forall x_\psi \in \R^{n_\psi}, \label{eq:target_contain} \\
& P_i k(t,x_G) \leq b_i, \forall (t,x) \in [0,T]\times \R^n, \nonumber \\
& \quad \quad  \text{s.t.} \ V(t,x) \leq \gamma + R^2, i = 1,...,n_p, \label{eq:control_saturate_thm1}
\end{align}
\endgroup
then the intersection of $\Omega^V_{0, \gamma}$ with the hyperplane $x_\psi = 0$ is an inner-approximation to $BRS(T,X_T,U,R,F_u(G,\Delta))$ under the control law $k$.
\end{theorem}

\begin{proof} 
	Since the dissipation inequality \eqref{eq:diss_basic} only holds on the local region $\Omega_{t,\gamma + R^2}^V$, we first need to prove that all the state trajectories starting from $\Omega_{0,\gamma}^V$ won't leave $\Omega_{t,\gamma + R^2}^V$ for all $t \in [0,T]$. This is proved by contradiction. Assume there exists a time instance $T_1 \in [0, T]$, $x_0 \in \Omega_{0,\gamma}^V$, such that a trajectory starting from $x(0) = x_0$ satisfies $V(T_1,x(T_1)) > \gamma + R^2$. Define $T_2 = \inf_{V(t,x(t))> \gamma+R^2} t$, and integrate \eqref{eq:diss_basic} over $[0, T_2]$:
	\begin{align}
	V(T_2, x(T_2)) - V(0,x(0))& \nonumber \\
	 + \int_0^{T_2} z(t)^\top M z(t) dt &\leq  \int_0^{T_2} d(t)^\top d(t) dt. \nonumber \\
	 \intertext{Apply $x(0) = x_0 \in \Omega_{0,\gamma}^V$ and $\Delta \in$ HardIQC$(\Psi,M)$ to show}
	V(T_2,x(T_2)) &\leq \gamma + \int_0^{T_2} d(t)^\top d(t) dt. \label{eq:proof1_a}
	\intertext{Next recall that $d$ is assumed to satisfy \eqref{eq:w_ass}:}
	\gamma + R^2 = V(T_2,x(T_2)) &< \gamma + R^2, \label{eq:proof1_b}
	\end{align}
which is a contradiction. As a result, $x(0) \in \Omega_{0,\gamma}^V$ implies $x(t) \in \Omega_{t,\gamma+R^2}^V$ for all $t \in [0,T]$, and thus $V(T,x(T))\leq \gamma+R^2$. Combining it with \eqref{eq:target_contain} shows that $\Omega_{0,\gamma}^V$ is an inner-approximation to the BRS of the extended system, and the intersection of  $\Omega_{0,\gamma}^V$ with $x_\psi = 0$ is an inner-approximation to $BRS(T,X_T,U,R,F_u(G,\Delta))$.
\end{proof}

To find a storage function $V$ and a control law $k$ satisfying the conditions of Theorem~\ref{thm:hardIQC}, we make use of sum-of-squares (SOS) programming. To do so, we restrict the decision variables to polynomials $V \in \R[(t,x)]$, $k \in \R^{n_u}[(t,x_G)]$, and make the following assumption.
\begin{assumption}
The nominal system $G$ given in \eqref{eq:uncertain_sys} has polynomial dynamics: $f \in \R^{n_G}[(x_G,w,d)]$, $g \in \R^{n_G \times n_u}[(x_G,w,d)]$, and $h \in \R^{n_v}[(x_G,w,d)]$. Therefore, $F$ and $H$ in \eqref{eq:extend_sys} are polynomials. $X_T$ is a semi-algebraic set: $X_T := \{x_G: p_x(x_G) \leq 0\}$, where $p_x \in \R[x_G]$ is provided.
\end{assumption}

In Example~\ref{ex:hardIQC}, we have seen that for each type of perturbation, any IQC defined by a properly chosen $\Psi$ and a $M$ drawn from the constraint set $\mathcal{M}$ is valid. Therefore, along with $V$ and $k$, we also treat $M \in \mathcal{M}$ as a decision variable. Assume $\mathcal{M}$ is described by linear matrix inequalities. Define $p_t : = t(T-t)$, which is nonnegative for all $t \in [0,T]$. By applying the generalized S-procedure \cite{Parrilo:00} to \eqref{eq:diss_basic} -- \eqref{eq:control_saturate_thm1}, and choosing the volume of $\Omega_{0,\gamma}^V$ as the objective function (to be maximized), we obtain the following optimization problem:
\begingroup
\allowdisplaybreaks
\begin{subequations}\label{eq:sos_hard}
\begin{align}
	\sup_{V,M,k,s_i} &\text{Volume}(\Omega_{0,\gamma}^V) \nonumber \\
	\text{s.t.} \ \ & V \in \R[(t,x)], k \in \R^{n_u}[(t,x_G)], M \in \mathcal{M}, \nonumber \\
	& -(\partial_t V + \partial_x V \cdot F \vert_{u=k} + z^\top M z - d^\top d) - s_1 p_t \nonumber \\
	& \quad \quad + (V - \gamma - R^2)s_2 \in \Sigma[(t,x,w,d)],\\
	& - s_3 p_x + V\vert_{t=T} - \gamma - R^2 \in \Sigma[x],\\
	&-(P_i k - b_i) - s_{4,i} p_t + (V - \gamma - R^2)s_{5,i} \nonumber \\
	&\quad \quad   \in \Sigma[(t,x)], \forall i = 1,...,n_p,
\end{align}
\end{subequations}
\endgroup
where polynomials decision variables $s_1, s_2 \in \Sigma[(t,x,w,d)]$, $(s_3-\epsilon) \in \Sigma[x]$, and $s_{4,i}, s_{5,i} \in \Sigma[(t,x)]$ are called S-procedure certificates or multipliers. The positive number $\epsilon$ ensures that $s_3$ is uniformly bounded away from 0. The optimization \eqref{eq:sos_hard} is a nonconvex SOS problem, since it is bilinear in two sets of decision variables, $V$ and $(k,s_2,s_{5,i})$.  Similar to \cite{Yin:TAC}, this noncovex optimization can be handled by alternating the search over these two sets of decision variables, since holding one set fixed and optimizing over the other results in a convex problem. The algorithm for solving \eqref{eq:sos_hard} is summarized in Algorithm~\ref{alg:alg1}, the $\gamma$-step of which treats $\gamma$ as a decision variable. By maximizing the value of $\gamma$, the volume of $\Omega_{0,\gamma}^{V^{j-1}}$ can be enlarged. The constraint \eqref{eq:levelset_grow} in the $V$-step enforces $\Omega_{0,\gamma^j}^{V^{j-1}} \subseteq \Omega_{0,\gamma^j}^{V^j}$. As proven in \cite{Yin:TAC}, the inner-approximation certified in one iteration contains the one certified in the previous iteration. A linear state feedback
for the linearization about the equilibrium point was
used to compute the initial iterate, $V^0$ \cite{Topcu:09}.
\begin{algorithm} [H]
	\caption{Iterative method for hard IQCs}
	\label{alg:alg1}
	\begin{algorithmic}[1]
		\Require{function $V^0$ such that constraints \eqref{eq:sos_hard} are feasible by proper choice of $s_i, k, \gamma,M$.}
		\Ensure{($k$, $\gamma$, $V$, $M$) such that with the volume of $\Omega_{0,\gamma}^V$ having been enlarged.}
		\For{$j = 1:N_{iter}$}
		\State $\boldsymbol{\gamma}$\textbf{-step}: decision variables $(s_i, k,\gamma,M)$.
			
			Maximize $\gamma$ subject to \eqref{eq:sos_hard}  using $V = V^{j-1}$. 
			
			This yields ($s_2^j, s_{5,i}^j, k^j$) and optimal reward $\gamma^j$.
		\State $\boldsymbol{V}\textbf{-step}$: decision variables $(s_0, s_1, s_3, s_{4,i}, V,M)$; 
			
			Maximize the feasibility subject to \eqref{eq:sos_hard} as well as 
			
			$s_0 \in \Sigma[x]$, and
			\begin{align}
			\ \ \ (\gamma^j - V\vert_{t=0}) +  (V^{j-1}\vert_{t=0} - \gamma^j) s_0 \in \Sigma[x], \label{eq:levelset_grow}
			\end{align}

			using ($\gamma = \gamma^j, s_2 = s_2^j, s_{5,i} = s_{5,i}^j, k=k^j$). This 
			
			yields $V^j$.
		\EndFor
	\end{algorithmic}
\end{algorithm}

\section{Extention to Actuator Uncertainty} \label{sec:actuator_uncert}
This section considers the case where the control inputs are subject to actuator uncertainty. In particular, consider the case where the input commanded by the controller is $u$ but the actual effect on the plant dynamics is the perturbed input $u_{pert}$.   For example, unmodeled actuator dynamics can be modeled as follows where $\Delta$ is a norm-bounded nonlinearity:
\begin{align}  
  u_{pert} = u + \Delta(u).
\end{align}
The input $v$ to $\Delta$ and the IQC filter output $z$ were previously defined (Equations \eqref{eq:expression_v} and \eqref{eq:z_expression}) to be independent of the control command $u$. However, the inclusion of the actuator uncertainty implies that $v$ and $z$ must now depend on $u$. 

This motivates the following generalization of the proposed method.  Assume the entire input vector $u$ is subject to the actuator uncertainty.  The perturbation input and IQC filter output are now given by the following modifications to Equations  \eqref{eq:expression_v} and \eqref{eq:z_expression}:
\begin{align}
   v(t) & = h(x_G(t),w(t),d(t),u(t)), \\
   z(t) & = H(x(t),w(t),d(t),u(t)). \label{eq:new_H}
\end{align}
A consequence of this generalization is that optimization over $k$ is nonconvex even when $V$ is fixed, since $z^\top M z$ in \eqref{eq:diss_basic} depends nonlinearly on $k$. A remedy is to introduce auxiliary state $\tilde x \in \R^{n_{u}}$ for the perturbed control input $u$, and to design a dynamic controller of the form
\begin{subequations}\label{eq:control_sys}
	\begin{align}
	\dot{\tilde x}(t) &= \tilde k(t,x_G(t),\tilde x(t)), \\
	u(t) &= \tilde x(t).
	\end{align}
\end{subequations}
where $ \tilde k : \R \times \R^{n_G} \times \R^{n_{u}} \rightarrow \R^{n_{u}}$ is to be determined. If we restrict the initial condition of $\tilde x$ to be zero: $\tilde x(0) = 0^{n_u}$, allow $\tilde k$ to depend on $\tilde x$, but not on $x_\psi$, and $V$ to depend on the new state $\tilde x$: $V: \R \times \R^{n}\times \R^{n_{u}} \rightarrow \R$, then the dissipation inequality becomes:
\begin{align}
&\partial_t V(t,x, \tilde x) + \partial_x V(t,x, \tilde x) \cdot F(x, w,d, \tilde x)  \nonumber \\
&\quad \quad + \partial_{\tilde x} V(t,x, \tilde x) \cdot \tilde k(t,x_G,\tilde x) +z^\top M z \leq d^\top d,  \nonumber \\
& \quad \quad \ \forall (t, x,\tilde x, w,d) \in  [0,T]\times \R^{n} \times \R^{ n_u} \times \R^{n_w} \times \R^{n_d}, \nonumber \\
&\quad \quad \ \text{s.t.} \ V(t,x, \tilde x) \leq \gamma + R^2. \label{eq:dissi_input_perturb}
\end{align}
The term $z^\top M z$ in \eqref{eq:dissi_input_perturb} is then nonlinear in the state variable $\tilde x$, rather than in the control law. The dissipation inequality is therefore bilinear in $V$ and $\tilde k$, and can be solved in a way similar to Algorithm~\ref{alg:alg1}. Next, we provide the theorem that incorporates actuator uncertainties. 
\begin{theorem} \label{thm:act_uncer}
	Let Assumption~\ref{ass:disturb} hold, and assume $\Delta \in$ HardIQC$(\Psi, M)$, with $\Psi$ and $M$ given. Given $X_T \subset \R^{n_G}$, $P \in \R^{n_p \times  n_u}$, $b \in \R^{n_p}$, $R > 0$, $F$ defined in \eqref{eq:extended_F}, $H$ defined in \eqref{eq:new_H}, $T > 0$, and $\gamma \in \R$, if there exists a $\mathcal{C}^1$ function $V: \R \times \R^{n}\times \R^{n_{u}} \rightarrow \R$, and control law $\tilde k : \R \times \R^{n_G} \times \R^{ n_{u}} \rightarrow \R^{n_{u}}$, such that \eqref{eq:dissi_input_perturb}, 
	\begin{align}
	&\{x_G : V(T,x,\tilde{x})\leq \gamma +R^2 \} \subseteq X_T, \nonumber \\
	& \quad \quad \quad \quad  \forall (x_\psi, \tilde x) \in \R^{n_\psi} \times U,  \\
	& P_i \tilde x \leq b_i, \forall (t,x,\tilde x) \in [0,T] \times \R^n \times \R^{n_u}, \nonumber \\
	&\quad \quad  \quad \quad \text{s.t.}  \ V(t,x,\tilde x) \leq \gamma + R^2, i = 1,...,n_p,
	\end{align}
	then the intersection of $\Omega_{0,\gamma}^V$ with the hyperplane $(x_\psi, \tilde x) = 0$ is an inner-approximation to $BRS(T,X_T,U,R,F_u(G,\Delta))$ under the control \eqref{eq:control_sys}.
\end{theorem}

The conditions of Theorem~\ref{thm:act_uncer} can be formulated as an SOS optimization similar to \eqref{eq:sos_hard}, and is omitted. The results go through, mainly with notation changes, when only a subset of the control inputs are perturbed by the uncertainty.

\section{Backward Reachability with soft IQCs}\label{sec:softIQC}
Previously we assumed $\Delta$ $\in$ HardIQC$(\Psi,M)$. However, many IQCs are specified in the frequency domain \cite{Megretski:97}, and an equivalent time domain representation results in a `soft IQC' as defined below.
\begin{definition}
Let $\Pi = \Pi^\sim \in \R\mathbb{L}_\infty^{(n_v + n_w)\times(n_v + n_w)}$ be given. A bounded, causal operator $\Delta : \mathcal{L}_{2}^{n_v} \rightarrow \mathcal{L}_{2}^{n_w}$ satisfies the \underline{frequency domain IQC} defined by the multiplier $\Pi$ if, for all $v \in \mathcal{L}_2^{n_v}$, and $w = \Delta(v)$,
\begin{align}
	\int_{-\infty}^\infty \bmat{\hat{v}(j \omega)\\\hat{w}(j\omega)}^* \Pi(j\omega) \bmat{\hat{v}(j \omega)\\\hat{w}(j\omega)} d\omega \ge 0, \nonumber 
\end{align}
where $\hat{v}$ and $\hat{w}$ are Fourier transforms of $v$ and $w$.
\end{definition}
Next, we provide the definition of the time domain soft IQC that is specified by $(\Psi, M)$.
\begin{definition}
Given $\Psi \in \R\mathbb{H}_{\infty}^{n_z \times (n_v + n_w)}$ and $M \in \mathbb{S}^{n_z}$. A bounded, causal operator $\Delta : \mathcal{L}_{2}^{n_v} \rightarrow \mathcal{L}_{2}^{n_w}$ satisfies the \underline{soft IQC} defined by $(\Psi,M)$ if, for all $v \in \mathcal{L}_{2}^{n_v}$, and $w = \Delta(v)$,
	\begin{align}
	\int_0^\infty z(\tau)^\top M z(\tau) d\tau \ge 0.
	\end{align} 
\end{definition}

Let $\Delta \in$ FreqIQC$(\Pi)$ and $\Delta \in$ SoftIQC$(\Psi, M)$ indicate that $\Delta$ satisfies corresponding frequency domain and time domain soft IQCs, respectively. Note that if $\Delta$ satisfies a time domain (hard or soft) IQC defined by $(\Psi, M)$, then $\Delta \in$ FreqIQC$(\Psi^\sim M \Psi)$. Conversely, any frequency domain multiplier $\Pi$ can be factorized (non-uniquely) as: $\Pi = \Psi^\sim M \Psi$ with $\Psi$ stable. By Parseval's theorem \cite{Kemin:1996}, $\Delta \in$ FreqIQC$(\Pi)$ implies $\Delta \in$ SoftIQC$(\Psi,M)$ for any such factorization. However, $\Delta \in$ FreqIQC$(\Pi)$ doesn't imply $\Delta \in$ HardIQC$(\Psi, M)$ in general. Hence, the library of IQCs specified in frequency domain can always be translated into soft IQCs, but not into hard IQCs. In addition, when both hard and soft factorizations exist, the latter is usually less restrictive. Therefore, it is helpful to incorporate soft IQCs in the analysis. Here, we provide one type of uncertainty and its corresponding frequency and time domain IQCs.
\begin{example}\label{ex:IQC_example_para}
	Consider the set of real constant parametric uncertainties: $w(t) = \Delta(v(t)) = \delta v(t)$, satisfying $\delta \leq \sigma$. From \cite{Megretski:97}, the frequency domain filter is chosen as $\Pi_\delta = \left[\begin{smallmatrix}\sigma^2\Pi_{11}(j\omega) & \Pi_{12}(j\omega) \\ \Pi_{12}^*(j\omega) & -\Pi_{11}(j\omega) \end{smallmatrix} \right]$, where $\Pi_{11}(j\omega)=\Pi_{11}^*(j\omega) \ge 0$ and $\Pi_{12}(j\omega)=-\Pi_{12}^*(j\omega)$ for all $\omega$. A soft IQC factorization for $\Pi_\delta$ is given by $\Psi= \left[\begin{smallmatrix} \Psi_{11}^{d,m} & 0 \\ 0 & \Psi_{11}^{d,m} \end{smallmatrix}\right]$, where $\Psi_{11}^{d,m}$ is defined in \eqref{eq:filter_choice}, and $M_{DG} = \left[\begin{smallmatrix} \sigma^2 M_{11} & M_{12} \\ M_{12}^\top & -M_{11} \end{smallmatrix} \right]$, where decision matrices are subject to $M_{11} = M_{11}^\top$, $M_{12} = -M_{12}^\top$, and $\Psi^{d,m\sim}_{11} M_{11}\Psi_{11}^{d,m} \ge 0$, which can be enforced by a KYP LMI \cite{RANTZER19967}. Notice that $\delta$ is a special case of the perturbation considered in Example \ref{ex:hardIQC} (a), and  thus $\delta \in$ HardIQC$(\Psi, M_D)$ as well. However, since $M_D$ is a special case of $M_{DG}$ with $M_{12}\equiv 0$, the analysis using $(\Psi, M_{DG})$ can be less conservative than using $(\Psi, M_D)$. 
\end{example}

Since soft IQCs hold over the infinite horizon, they cannot be incorporated in the analysis based on a finite-horizon dissipation inequality directly. To alleviate this issue, we use the following lemma which provides lower bounds for soft IQCs over all finite horizons, and thus allows for soft IQCs in the finite horizon reachability analysis. Let $\Pi = \left[\begin{smallmatrix}\Pi_{11} & \Pi_{12}\\\Pi_{12}^\sim & \Pi_{22} \end{smallmatrix}\right]$ be a partition conformal with the dimensions of $v$ and $w$.
\begin{lemma} \label{lemma:IQClowerbound}(\cite{Fetzer:18})
	Let $\Psi \in \R\mathbb{H}_\infty^{n_z \times (n_v + n_w)}$ and $M \in \mathbb{S}^{n_z}$ be given. Define $\Pi := \Psi^{\sim} M \Psi$. If $\Pi_{22}(j\omega) < 0 \ \forall \omega$, then 
	\begin{itemize}
		\item $D_{\psi 2}^\top M D_{\psi 2}<0$ and there exists a  $Y_{22} \in \mathbb{S}^{n_\psi}$ satisfying
		\begin{align}
		 KYP(Y_{22}, A_\psi,B_{\psi 2}, C_\psi, D_{\psi 2},M) < 0. \label{eq:exact_kyp}
		 \end{align}
		\item If $\Delta \in$ SoftIQC$(\Psi,M)$ then for all $t \ge 0$, $v \in \mathcal{L}_{2}^{n_v}[0,\infty]$, $w = \Delta(v)$, and $Y_{22} \in \mathbb{S}^{n_\psi}$ satisfying \eqref{eq:exact_kyp},
		\begin{align}
		\int_{0}^t z(\tau)^\top M z(\tau)d \tau \ge -x_\psi(t)^\top Y_{22}x_\psi(t).
		\end{align}
	\end{itemize}
\end{lemma}
Based on this lemma, the following theorem provides a BRS inner-approximation for $F_u(G,\Delta)$ with $\Delta \in$ SoftIQC$(\Psi, M)$, also allowing for actuator uncertainties.
\begin{theorem}\label{thm:softIQC}
	Let Assumption~\ref{ass:disturb} hold, and assume $\Delta \in$ SoftIQC$(\Psi, M)$, with $\Psi$ and $M$ given. Given $X_T \subset \R^{n_G}$, $P \in \R^{n_p \times  n_u}$, $b \in \R^{n_p}$, $R > 0$, $F$ defined in \eqref{eq:extended_F}, $H$ defined in \eqref{eq:new_H}, $T > 0$, and $\gamma \in \R$, if there exists a $\mathcal{C}^1$ function $V: \R \times \R^{n}\times \R^{n_{u}} \rightarrow \R$, a matrix $Y_{22} \in \mathbb{S}^{n_\psi}$ satisfying \eqref{eq:exact_kyp}, and control law $\tilde k : \R \times \R^{n_G} \times \R^{n_{u}} \rightarrow \R^{n_{u}}$, such that
	\begin{subequations}
	\begin{align}
	&\partial_t V(t,x, \tilde x) + \partial_x V(t,x, \tilde x) \cdot F(x, w,d, \tilde x)  \nonumber \\
	&\quad  + \partial_{\tilde x} V(t,x, \tilde x) \cdot \tilde k(t,x_G,\tilde x) +z^\top M z \leq d^\top d,  \nonumber \\
	& \quad  \ \forall (t, x,\tilde x, w,d) \in  [0,T]\times \R^{n} \times \R^{n_u} \times \R^{n_w} \times \R^{n_d}, \nonumber \\
	&\quad  \ \text{s.t.} \ \mathcal{V}(t,x,\tilde x) \leq \gamma + R^2,  \label{eq:dissi_input_perturb_softIQC}\\
	&\{x_G : \mathcal{V}(T,x,\tilde x) \leq \gamma +R^2 \} \subseteq X_T, \nonumber  \\
	& \quad  \forall (x_\psi, \tilde x) \in \R^{n_\psi} \times U, \label{eq:target_contain_soft} \\
	& P_i \tilde x \leq b_i, \forall (t,x,\tilde x) \in [0,T] \times \R^n \times \R^{n_u}, \nonumber \\
	&\quad \text{s.t.}  \ \mathcal{V}(t,x,\tilde x)  \leq \gamma + R^2, i = 1,...,n_p,
	\end{align}
	\end{subequations}
 where $\mathcal{V} = V - x_\psi^\top Y_{22} x_\psi$, then the intersection of $\Omega_{0,\gamma}^V$ with the hyperplane $(x_\psi, \tilde x) = 0$ is an inner-approximation to $BRS(T,X_T,U,R,F_u(G,\Delta))$  under the control \eqref{eq:control_sys}.
\end{theorem}
\begin{proof}
	Similar to the proof of Theorem~\ref{thm:hardIQC}, it follows by contradiction that $(x(0),\tilde x(0)) \in \Omega_{0,\gamma}^V$ implies $(x(t),\tilde x(t)) \in \Omega_{t,\gamma+R^2}^{\mathcal{V}}$, for all $t \in [0, T]$. Therefore, we are able to integrate \eqref{eq:dissi_input_perturb_softIQC} over $[0,T]$:
	\begin{align}
	&V(T,x(T),\tilde x(T)) - V(0,x(0),\tilde x(0)) \nonumber \\
	& \quad \quad + \int_{0}^T z(t)^\top M z(t)dt \leq \int_{0}^T  d(t)^\top d(t)dt. \nonumber \\
	\intertext{Use $(x(0),\tilde x(0)) \in \Omega_{0,\gamma}^V$ and $\norm{d}_{2,[0,T]} < R$ to show}
	&V(T,x(T),\tilde x(T)) + \int_{0}^T z(t)^\top M z(t)dt < \gamma + R^2. \nonumber
	\intertext{Next it follows from $\Delta \in$ SoftIQC$(\Psi,M)$ and Lemma~\ref{lemma:IQClowerbound} that}
	& V(T,x(T),\tilde x(T)) - x_\psi(T)^\top Y_{22} x_\psi(T)  < \gamma + R^2 \label{eq:softpf2}.
	\end{align}
	Combining \eqref{eq:softpf2} with \eqref{eq:target_contain_soft}, it holds $x_G(T) \in X_T$ for all $(x(0),\tilde x(0)) \in \Omega_{0,\gamma}^V$. Therefore, the intersection of $\Omega_{0,\gamma}^V$ with $(x_\psi, \tilde x) = 0$ is an inner-approximation to $BRS(T,X_T,U,R,F_u(G,\Delta))$.
\end{proof}

 Similar to \eqref{eq:sos_hard}, we can formulate SOS optimization using the constraints of Theorem~\eqref{thm:softIQC}
\begingroup
\allowdisplaybreaks
\begin{subequations}\label{eq:sos_soft}:
	\begin{align}
	\sup_{V,M,Y_{22}, \tilde k,s_i} &\text{Volume}(\Omega_{0,\gamma}^V) \nonumber \\
	\text{s.t.} \ \ \ & V \in \R[(t,x,\tilde x)], \tilde k \in \R^{ n_u}[(t,x_G,\tilde x)], \nonumber \\
	&M \in \mathcal{M} \ \text{and} \ Y_{22} \in \mathbb{S}^{n_\psi} \ \text{satisfy \eqref{eq:exact_kyp}}, \nonumber \\
	& -(\partial_t V + \partial_x V \cdot F \vert_{u= \tilde x} + \partial_{\tilde x} V \cdot \tilde k \nonumber \\
	& \quad + z^\top M z - d^\top d)  + (\mathcal{V} - \gamma - R^2)s_2   \nonumber \\
	& \quad - s_1 p_t \in \Sigma[(t,x,\tilde x,w,d)],\\
	& - s_3 p_x + \mathcal{V}\vert_{t=T} - \gamma - R^2 \nonumber \\
	&\quad  + \Sigma_{i=1}^{n_p} (P_i \tilde x - b_i)s_{6,i} \in \Sigma[(x,\tilde x)],\\
	&-(P_i \tilde x - b_i) + (\mathcal{V} - \gamma - R^2)s_{5,i}  \nonumber \\
	&\quad  - s_{4,i} p_t  \in \Sigma[(t,x,\tilde x)],  i = 1,...,n_p, 
	\end{align}
\end{subequations}
\endgroup
where $s_1, s_2 \in \Sigma[(t,x,\tilde x, w,d)]$, $(s_3-\epsilon), s_{6,i} \in \Sigma[(x,\tilde x)]$ and $s_{4,i}, s_{5,i} \in \Sigma[(t,x,\tilde x)]$. The optimization \eqref{eq:sos_soft} is bilinear in $(V, Y_{22})$ and $(s_2, s_{5,i}, \tilde k)$. Similar to Algorithm~\ref{alg:alg1}, Algorithm~\ref{alg:alg2} tackles \eqref{eq:sos_soft} by decomposing it into convex subproblems, and it also guarantees the improvement of the quality of the inner-approximation through iterations. $Y_{22}^0=0^{n_\psi}$ and a $M^0 \in \mathcal{M}$ can be used as initializations.
\begin{algorithm} [H]
	\caption{Iterative method for soft IQCs}
	\label{alg:alg2}
	\begin{algorithmic}[1]
		\Require{$V^0$, $M^0$ and $Y_{22}^0$ such that constraints \eqref{eq:sos_soft} are feasible by proper choice of $s_i,\tilde k, \gamma$.}
		\Ensure{($\tilde k$, $\gamma$, $V$, $M$, $Y_{22}$) such that with the volume of $\Omega_{0,\gamma}^V$ having been enlarged.}
		\For{$j = 1:N_{iter}$}
		\State $\boldsymbol{\gamma}$\textbf{-step}: decision variables $(s_i,  \tilde k, \gamma)$.
		
		Maximize $\gamma$ subject to \eqref{eq:sos_soft} using $V = V^{j-1}$,
		
		$M = M^{j-1}$ and $Y_{22} = Y_{22}^{j-1}$. 
		
		This yields ($s_2^j, s_{5,i}^j, \tilde k^j$) and optimal reward $\gamma^j$.
		\State $\boldsymbol{V}\textbf{-step}$: $(s_0, s_1, s_3, s_{4,i}, s_{6,i}, V,M, Y_{22})$ are  decision 
		
		variables. Maximize the feasibility subject to \eqref{eq:sos_soft} 
		
		as well as $s_0 \in \Sigma[(x,\tilde x)]$, and
		\begin{align}
		& (\gamma^j - V\vert_{t=0})  +  (V^{j-1}\vert_{t=0} - \gamma^j) s_0 \in \Sigma[(x,\tilde x)], \nonumber 
		\end{align}
		
		using $\gamma = \gamma^j, s_2 = s_2^j, s_{5,i} = s_{5,i}^j, \tilde k=\tilde k^j$. 
		
		This  yields $V^j$, $M^j$ and $Y_{22}^j$.
		\EndFor
	\end{algorithmic}
\end{algorithm}

\section{Numerical Examples}\label{sec:Numeric_ex}
In the following examples, the SOS optimization problem is formulated using the SOS module in SOSOPT \cite{Pete:13} on MATLAB, and solved by the SDP solver MOSEK \cite{Mosek:17}.
\subsection{Generic Transport Model (GTM) Example}
The GTM is a remote-controlled $5.5\%$ scale commercial aircraft \cite{Murch:07}. The longitudinal dynamics are approximated by a cubic degree polynomial model provided in \cite{Chakraborty:2011}:
\begingroup
\allowdisplaybreaks
\begin{align}
\dot{x}_1 &= - 1.492 x_1^3 + 4.239 x_1^2 + 0.003 x_1 x_2
+ 0.006 x_2^2 \nonumber \\
&  \quad - 3.236 x_1 + 0.923 x_2 + (0.240 x_1 - 0.317)u, \nonumber \\
\dot{x}_2 &= -7.228 x_1^3 + 1.103 x_2^3 + 18.365 x_1^2 
- 45.339 x_1 \nonumber \\
& \quad - 4.373 x_2 + (41.505 x_1 - 59.989) u, \nonumber
\end{align}
\endgroup
where $x_G = [x_1;x_2]$ is the state, $x_1$ is the angle of attack (rad), $x_2$ is the pitch rate (rad/s), and the control input $u$ is the elevator defection (rad). Assume the control input $u$ generated by the controller $C$ is corrupted by an additive uncertainty $\Delta$ exerted on the actuator, as shown in Fig.~\ref{fig:GTMdiagram}. The actual signal that goes into the elevator channel is $u_{elev} := u + w$, where $w$ is the output of $\Delta$. 
\begin{figure}[h]
	\centering
	\includegraphics[width=0.3\textwidth]{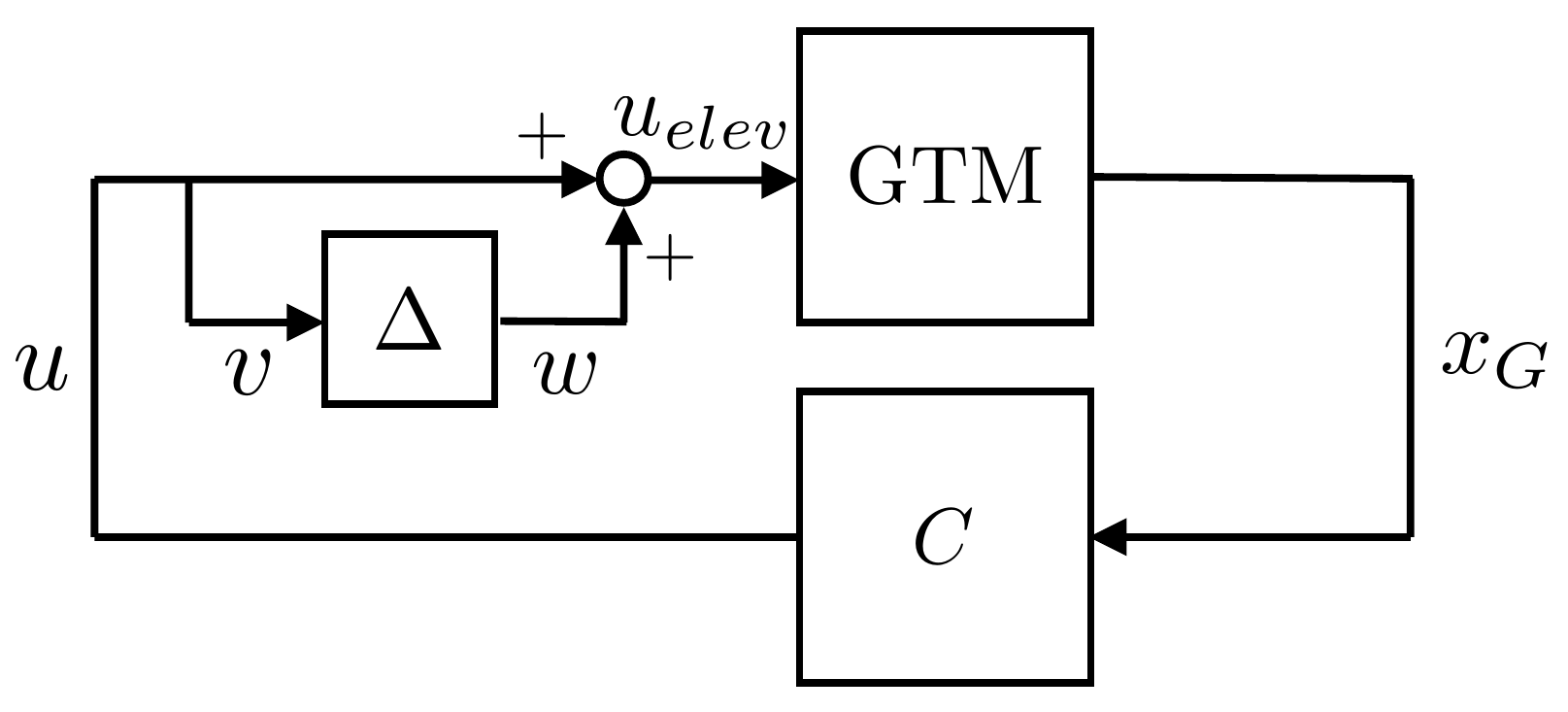}
	\caption{The diagram of the GTM with input perturbation}
	\label{fig:GTMdiagram}    
\end{figure}

\subsubsection{Sector IQCs} Assume that $\Delta$ lies within the sector $[\alpha, \beta]$, where $\alpha = 0$, and $\beta = 0.2$. The filter $\Psi$ and constraint set $\mathcal{M}$ given below define a hard IQC:
\begin{align}
\Psi = I_2, \ \mathcal{M} = \left\{\lambda \bmat{-2 \alpha \beta & \alpha + \beta \\ \alpha + \beta & -2}: \lambda \in \Sigma[(x_G,\tilde{x},w)] \right\}, \nonumber
\end{align}
where $\lambda$ is a polynomial decision variable, which introduces more freedom to the optimization.

Take the target set as $X_T = \{x_G| x_G^\top x_G \leq (\pi/27)^2 \}$ (shown in Fig.~\ref{fig:GTM} with red solid curve), and assume the actuator limit on $u$ is $|u(t)| \leq 0.261$ rad. Degree-4 polynomial storage functions are used to compute two inner-approximations on time horizons [0, 1 sec] and [0, 2 sec], which correspond to the blue dashed curve and black dotted curve in Fig.~\ref{fig:GTM}, respectively. Solid curves with crosses represent simulation trajectories starting from the inner-approximation with time horizon [0, 2 sec] in the presence of actuator uncertainty, and crosses represent different initial conditions. In Fig.~\ref{fig:GTM_control}, the simulations of control inputs for different initial conditions are shown. We note that they are all within the control limits during the time horizon.
\begin{figure}[h]
	\centering
	\includegraphics[width=0.35\textwidth]{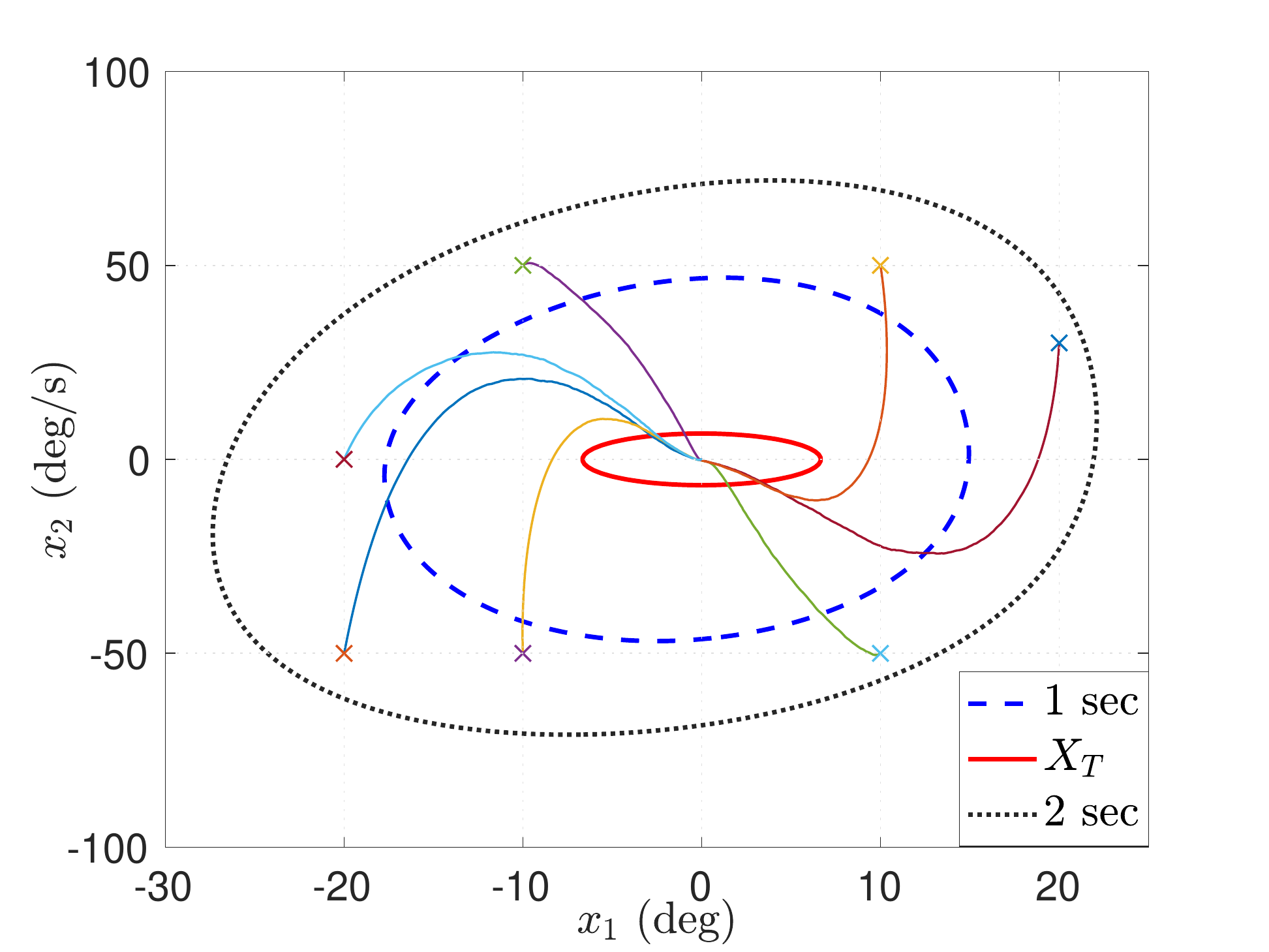}
	\caption{Simulation trajectories, and inner-approximations of the GTM example with the sector IQC for two time horizons.}
	\label{fig:GTM}    
\end{figure}
\vspace{-0.3in}
\begin{figure}[h]
	\centering
	\includegraphics[width=0.35\textwidth]{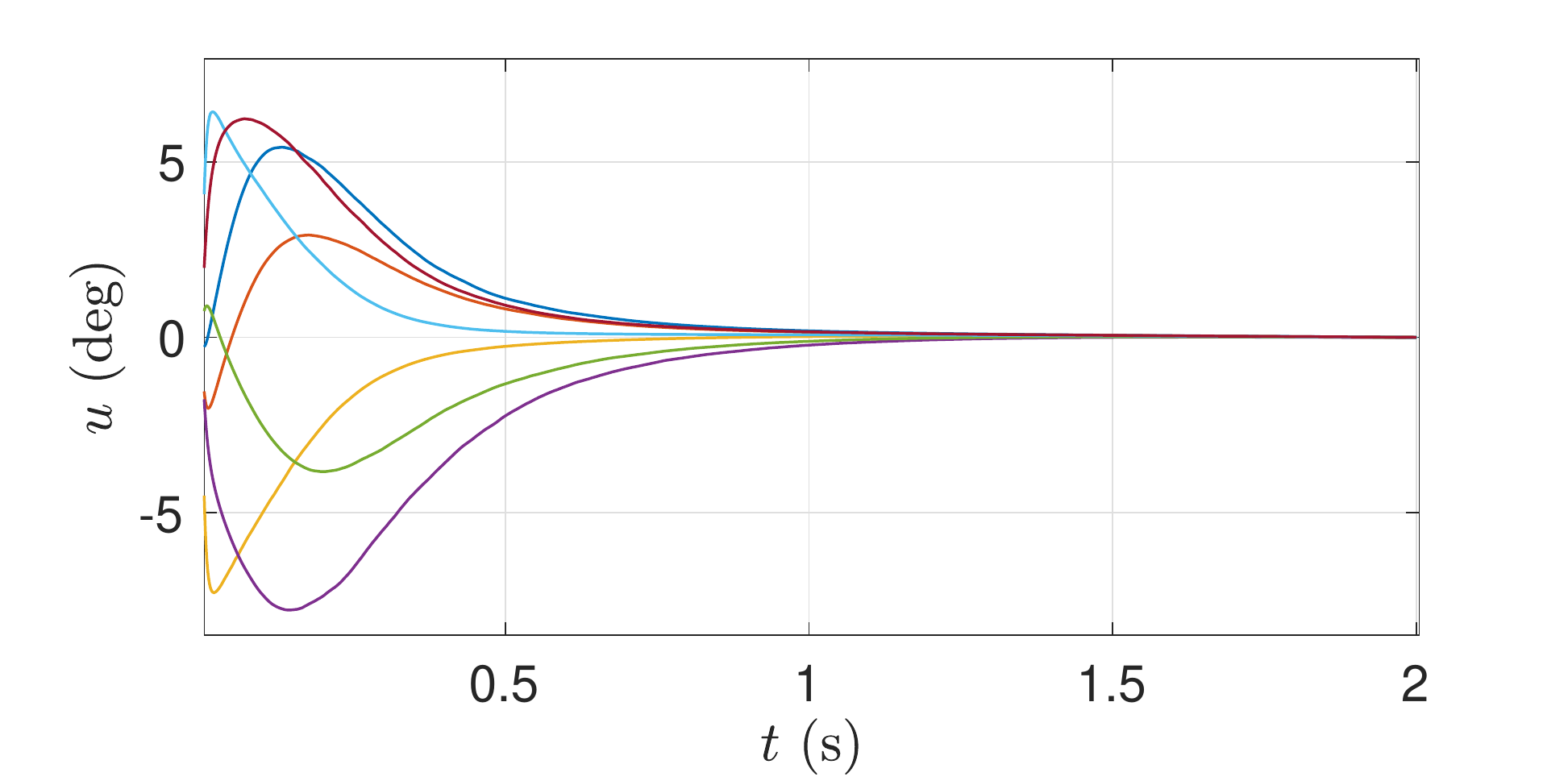}
	\caption{Simulations of control inputs}
	\label{fig:GTM_control}    
\end{figure}

\subsubsection{Hard and soft IQCs}
This time we assume that the perturbation $\Delta$ in Fig.~\ref{fig:GTMdiagram} is a time invariant parametric uncertainty: $w(t) = \Delta(v(t)) = \delta v(t)$, with $\delta \in \R$, $|\delta| \leq 0.2$. Therefore, the actual signal that goes into the elevator channel is $u_{elev} = u+w = (1+\delta) u$. As discussed in Example~\ref{ex:IQC_example_para}, $\delta$ satisfies both HardIQC$(\Psi,M_D)$ and SoftIQC$(\Psi,M_{DG})$. The backward reachability is performed using both kinds of IQCs. In both cases, we use the same filter $\Psi$, and choose $\Psi_{11}^{d,m}$ from \eqref{eq:filter_choice} with $m=10$ and $d=1$. Therefore, $\Psi$ introduces two filters states $x_\psi \in \R^{2}$ to the extended system. Take the time horizon as [0, 2 sec], and use the same target set and actuator limits from the previous example.

In Fig.~\ref{fig:GTM_soft_hard}, the inner-approximations computed using the hard and soft IQCs are shown with the dashed purple curve, and the dash-dotted black curve. We see that with soft IQC we are able to certify a larger inner-approximation. This is because the soft IQC has richer knowledge of the time invariant parametric uncertainty than the hard IQC.
\begin{figure}[h]
	\centering
	\includegraphics[width=0.35\textwidth]{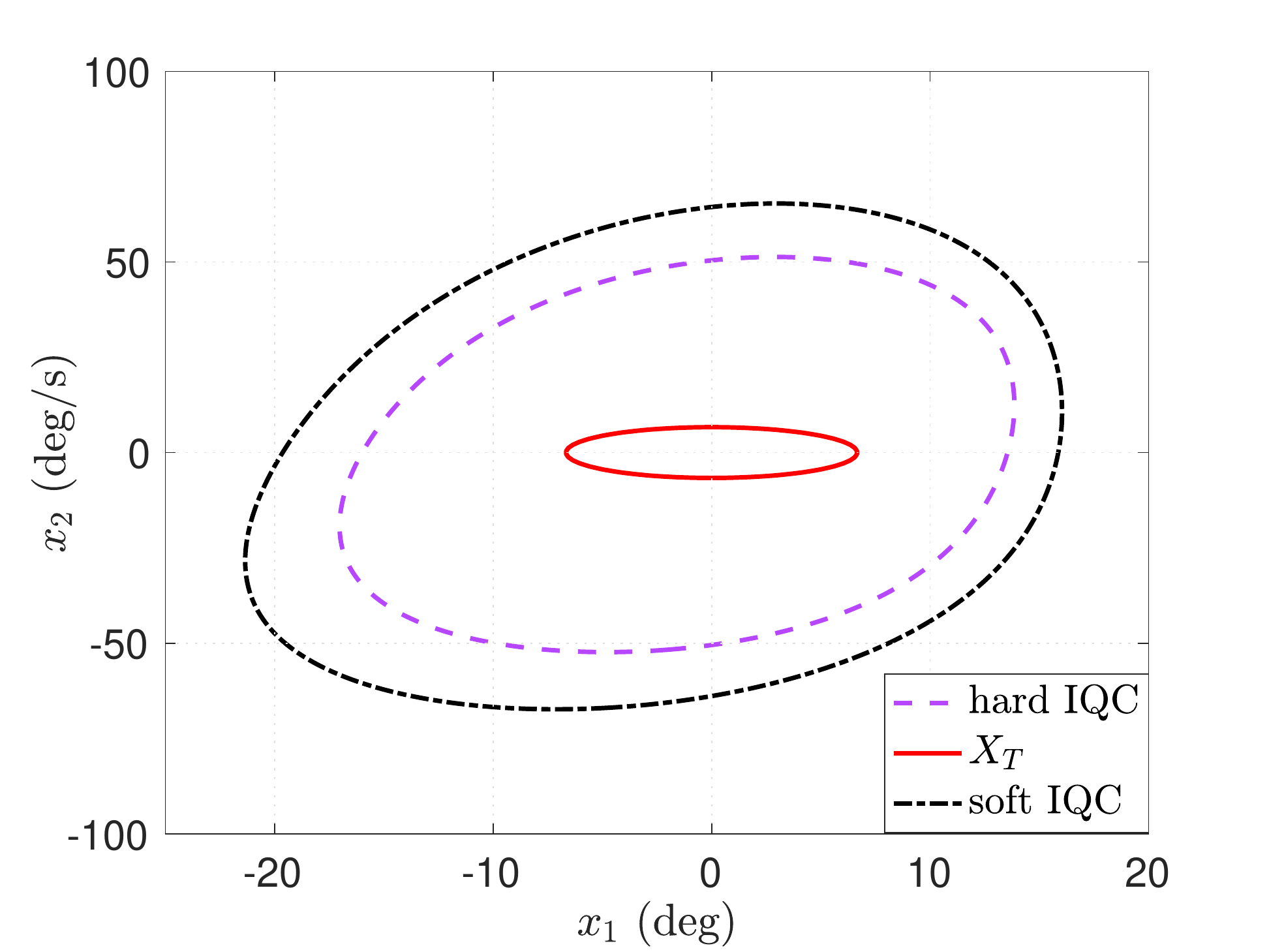}
	\caption{Inner-approximations with soft and hard IQCs}
	\label{fig:GTM_soft_hard}    
\end{figure}

\subsection{Quadrotor Example} \label{sec:quadrotor}
Consider the following 6-state planar quadrotor dynamics from \cite{Mitchell:19, Bouffard:12}:
\begingroup
\allowdisplaybreaks
\begin{align}
\dot{x}_1 &= x_3, \nonumber \\
\dot{x}_2 &= x_4, \nonumber \\
\dot{x}_3 &= u_1 K \sin(x_5), \nonumber \\
\dot{x}_4 &= u_1 K \cos(x_5) - g_n, \nonumber \\
\dot{x}_5 &= x_6, \nonumber \\
\dot{x}_6 &= -d_0 x_5 - d_1 x_6 + n_0 u_2, \nonumber
\end{align}
\endgroup
where $x_1$ to $x_6$ represent horizontal position (m), vertical position (m), horizontal velocity (m/s), vertical velocity (m/s), roll (rad), and roll velocity (rad/s), respectively. $u_1$ and $u_2$ represent total thrust and desired roll angle. Control saturation limits are $u_1(t) \in [-1.5, 1.5]+g_n/K$, and $u_2(t) \in [-\pi/12, \pi/12]$. Values for the constants are: $g_n = 9.8$, $K = 0.89/1.4$, $d_0 = 70$, $d_1 = 17$, and $n_0 = 55$. 

The control objective of this example is to design controllers for $u_1$ and $u_2$ to maintain the trajectories of the quadrotor starting from the BRS to stay within the safe set $X_t$ during the time horizon $[0, T]$ with $T = 2$. $X_t$ is given as $X_t = \{x_G: x_G^\top N x_G  \leq 1 \}$, where $N = diag(1/1.7^2, 1/0.85^2, 1/0.8^2, 1/1^2, 1/(\pi/12)^2, 1/(\pi/2)^2)$. $\sin(x_5)$ is approximated by $(-0.166 x_5^3+x_5)$ and $\cos(x_5)$ is approximated by $(-0.498 x_5^2+1)$, using least squares regression for $x_5 \in [-\pi/12, \pi/12]$. The validity of this bound on $x_5$ is guaranteed by the state constraint $X_t$. Assume that the control input $u_2$ is perturbed by an additive norm-bounded nonlinearity $\norm{\Delta}_{2\rightarrow 2,[0,T]} \leq 0.2$, which introduces one auxiliary state $\tilde x$ to the analysis. We use the hard IQC discussed in Example~\ref{ex:hardIQC}(b) with a fixed filter $\Psi$ and search for $M$ over the constraint set given in \eqref{eq:M_nonlinearity}. Inner-approximations to the BRS are computed using both degree-2 and degree-4 polynomial storage functions, with computation time of $1.1\times 10^3$ and $3.6\times 10^4$ seconds.

Fig.~\ref{fig:quadrotor} shows the projections of the resulting inner-approximations. The one computed using degree-2 storage function is shown with the solid magenta curve, and the one computed using degree-4 storage function is shown with the red dash-dotted curve. The projections of $X_t$ are shown with the blue solid curves.
\begin{figure}[h]
	\centering
	\includegraphics[width=0.48\textwidth]{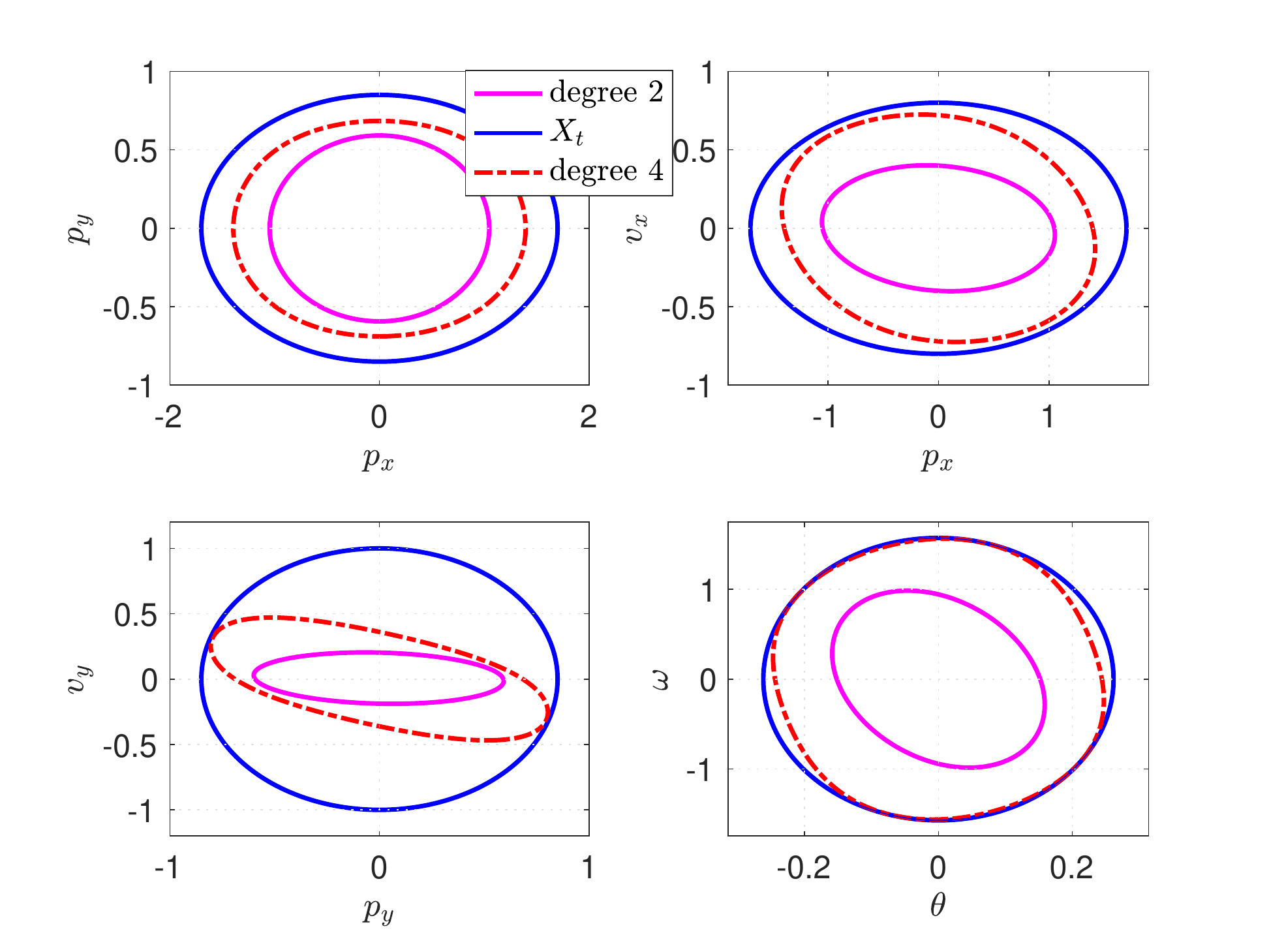}
	\caption{Inner-approximations to the BRS of quadrotor example with norm-bounded nonlinearity}
	\label{fig:quadrotor}    
\end{figure}

\section{Conclusions}\label{sec:conclusion}
In this paper, a method for computing robust inner-approximations to the BRS and robust control laws is proposed for uncertain nonlinear systems,  modeled as an interconnection of the nominal system $G$ and the perturbation $\Delta$. The proposed framework merges dissipation inequalities and IQCs, with both hard and soft factorizations. The use of IQCs enabled us to address a large class of perturbations, including uncertain time delay and unmodeled dynamics. The generalized S-procedure and sum-of-squares programming are used to derive computational algorithms. Finally, the effectiveness of the method is illustrated on uncertain nonlinear systems, including a 6-state quadrotor examples with actuator uncertainties.

\bibliographystyle{IEEEtran}
\bibliography{IEEEabrv,bibfile}

\end{document}